\documentclass[a4paper,10pt,oneside]{article}
\usepackage[top=2.50cm, bottom=0.5cm,left=2.5cm, right=1.5cm]{geometry}
\usepackage{amssymb,amsmath,amsthm}
\usepackage[numbers]{natbib}
\usepackage{amsmath}
\date{}
\usepackage{graphicx}
\usepackage{caption}
\usepackage{subcaption}
\usepackage{tabularx}
\usepackage{multirow}
\usepackage{float}
\usepackage[colorlinks]{hyperref}
\usepackage{multicol}
\usepackage{adjustbox}
\usepackage{colortbl}
\usepackage{pdflscape}
\usepackage{algorithm}
\usepackage{algorithmic}
\usepackage{epstopdf}
\usepackage[affil-it]{authblk}
\title{\bf{Record-based transmuted log-logistic distribution: Properties, simulation, and applications to petroleum rock and reactor pump data}}
\author{Caner Tan{\i}\c{s}}\affil{Department of Statistics, \c{C}ank{\i}r{\i} Karatekin University, \c{C}ank{\i}r{\i}, Turkey; canertanis@karatekin.edu.tr}

\newtheorem{theorem}{Theorem}

\begin{document}
\maketitle
\begin{abstract}
This study aims to introduce a new lifetime distribution, called the record-based transformed log-logistic distribution, to the literature. We obtain this distribution using a record-based transformation map based on the distributions of upper record values. We explore some mathematical properties of the suggested distribution, namely the quantile function, hazard function, moments, order statistics, and stochastic ordering. We discuss the point estimation via seven different methods such as maximum likelihood, least squares, weighted least squares, Anderson-Darling, Cramer-von Mises, maximum product spacings, and right tail Anderson Darling. Then, we perform a Monte Carlo simulation study to evaluate the performances of these estimators. Also, we present two practical data examples, reactor pump failure and petroleum rock data to compare the fits of the proposed distribution with its rivals. As a result of data analysis, we conclude that the best-fitted distribution is the record-based transmuted log-logistic distribution for reactor pump failure and petroleum rock data sets.
\end{abstract}


\section{Introduction}
Lifetime distributions are widely used in data analysis and modeling for future prediction and statistical inference in many fields. Many popular life-time distributions such as Weibull, gamma, normal, and exponential are used in modeling data obtained in many fields such as engineering, agriculture, biology, chemistry, medicine, economics, and social sciences. Considering both the increasing variety of data and the speed of the current data flow, current lifetime distributions are not sufficient for real-world data modeling and other statistical inferences. In this case, many researchers have aimed to introduce new lifetime distributions to the literature with some modifications based on existing distributions. For instance, the exponential distribution is known to be able to model the lifetime of a mechanical part. If a more flexible distribution that rivals the exponential distribution, that is, a distribution that has the potential to model data with more hazard rate shapes, is introduced to the literature, it will be seen that it is superior to the exponential distribution in modeling this type of data. New methods have been proposed in the literature for years to obtain new lifetime distributions. One of these methods is the transmutation map proposed by X. This method is based on the distribution of the first two order statistics.  The transmutation map is summarized below.

Let $X_1$ and $X_2$ be independent and identically distributed random variables with cumulative distribution function (CDF) $G\left(.\right)$ and probability distribution function (PDF) $g\left(.\right)$, and $X_{1:n} $, $X_{2:n} $ be the first two order statistics corresponding to the sample.

 Let us define random variable $U$ by 
 \begin{equation*}
 \begin{array}{l}
 U\overset{d}{=}X_{1:2},\mbox{ with probability }p,  \\ 
 U\overset{d}{=}X_{2:2},\mbox{ with probability }1-p ,%
 \end{array}%
 \end{equation*}
 and corresponding CDF is
\begin{eqnarray}
\label{transmuted: cdf}
F_{U}\left( x\right)  &=&p P\left( {X_{1:2}\leq x}\right) +\left( {1-p }%
\right) P\left( {X_{2:2}\leq x}\right)   \notag \\
&=&p \left( {1-\left( {1-\left( {G\left( x\right) }\right) ^{2}}\right) }%
\right) +\left( {1-p }\right) G^{2}\left( x\right)   \notag \\
&=&2p G\left( x\right) +\left( {1-2p }\right) G^{2}\left( x\right) .
\end{eqnarray}
 where $p \in (0,1) $

Substituting $p =\frac{1+\lambda }{2}$ in Eq. (\ref{transmuted: cdf}), the CDF and
corresponding the PDF are 

\begin{equation}
\label{eq2}
F\left(x\right)=\left(1+\lambda\right)G\left(x\right)-\lambda \left(G\left(x\right)\right)^2,
\end{equation}
and
\begin{equation}
\label{eq3}
f\left(x\right)=\left(1+\lambda\right)g\left(x\right)-2\lambda G\left(x\right)g\left(x\right),
\end{equation}
respectively, where $\lambda \in \left[-1,1\right]$.
The transmutation method, and specially the rank transmutation map, allows for the systematic alteration of a baseline distribution to generate more flexible families while retaining analytical tractability. Many authors have explored this idea to propose transmuted versions of well-known distributions, leading to a rich literature that demonstrates their superior performance in modeling lifetime, reliability, and biological data \cite{granzotto2015transmuted}. For example, transmuted distributions such as the Transmuted log-logistic \cite{aryal2013transmuted}, Transmuted half logistic \cite{samuel2019study}, and transmuted logistic exponential \cite{adesegun2023transmuted} have been shown to outperform their baseline distributions in goodness-of-fit and inferential flexibility. Also, \cite{louzada2016transmuted} suggested transmuted log-logistic regression model.

Similar to the family of distributions based on the distribution of the first two order statistics defined in Eq. (\ref{transmuted: cdf}), \cite{balakrishnan2021record} proposed a novel family of distributions based on the first two upper record statistics called record-based transmutation map (RBTM). The RBTM is given as follows:

Let $X_1 $ and $X_2 $  be a random sample with two sizes from the distribution with the CDF $G(.)$ and PDF $g(.)$, and $X_{U\left( 1 \right)} $ and $X_{U\left( 2 \right)} $ 
be upper records associated with the sample.
 
 Let us define a random variable $Y$ 
\[
\begin{array}{l}
 Y\overset{d}{=} \mbox{ }X_{U\left( 1 \right)} ,\mbox{ 
with probability }p _1, \\ 
 Y\overset{d}{=} \mbox{ }X_{U\left( 2 \right)} ,\mbox{ 
with probability }p _2, \\ 
 \end{array}
\]
where $U_{\left( n\right) }=\min \left\{ i:i>U\left( n-1\right)
,X_{i}>X_{U\left( n-1\right) }\right\} \left\{ U_{\left( n\right) }\right\}
_{n=1}^{\infty }$ denotes upper record times and $\left\{ X_{U\left(
n\right) }\right\} _{n=1}^{\infty }$ denotes the corresponding record sequence \cite{balakrishnan2021record},\citet{arnold2008first}, $p_{1}+p_{2}=1$ and $Y$ refers to a random variable having  record-based transmuted distribution. In this regard, the CDF of $Y$ is 
\begin{eqnarray}
F_{Y}\left( x\right)  &=&p_{1}P\left( {X_{U\left( 1\right) }\leq x}\right)
+p_{2}P\left( {X_{U\left( 2\right) }\leq x}\right)   \nonumber \\
&=&G\left( x\right) +p\left[ {\left( {1-G\left( x\right) }\right) \log
\left( {1-G\left( x\right) }\right) }\right], 
\end{eqnarray}
where $p\in \left( {0,1} \right)$. The corresponding PDF is  
 
\begin{equation}
\label{eq9}
f_Y \left( x \right)=g\left( x \right)\left[ {1+p\left( {-\log \left( 
{1-G\left( x \right)} \right)-1} \right)} \right].
\end{equation}

The concept of \emph{record-based} or \emph{lower record type} transmuted distributions has recently emerged, allowing the authors to generate new lifetime distributions by means of the distributions of upper and lower record values via the RBTM \cite{tanis2021transmuted, tanis2022record}. In recent years, various special cases of this family of distribution have been proposed via the RBTM such as Weibull \cite{tanics2022record}, power Lomax \cite{sakthivel2022record}, Lindley \cite{tanics2024new}, generalized linear exponential \cite{arshad2024record}, unit Omega \cite{pathak2024record}, Burr X \cite{alrweili2025statistical}.

The aim of this study is to propose a more flexible distribution for data modeling than the log-logistic distribution by adopting the log-logistic distribution as the base distribution. Thus, there will be a viable alternative for data modeling where the log-logistic distribution is not sufficient. The distribution proposed in our study was obtained by means of RBTM. At the same time, this new distribution, which will be a competitor to the distributions generated using RBTM, will fill an important gap in the literature, as it has never been studied before.

The remainder of this paper is organized as follows. Sections \ref{Sec: RBTLL} and \ref{Sec: propRBTLL} introduces a new submodel of the record-based transmuted family of distributions based on log-logistic distribution and some statistical properties, respectively. In Section \ref{Estimation methods}, seven estimators are suggested to estimate the parameters of the proposed distribution. Then, a simulation study is considered to observe the performances of these estimators under different sample sizes and parameter settings. To generate a random sample from the introduced distribution, we provide an algorithm in Section \ref{Sec: simulation}. In Section \ref{Sec: realdata}, two real-world data examples associated with reactor pump failure and petroleum rock are presented. Finally, the concluding remarks are given in Section \ref{Sec: Conc}.

\section{Record-based transmuted Log-Logistic distribution}
\label{Sec: RBTLL}
The log-logistic distribution of the cumulative distribution function (CDF) and the probability density function (PDF) are 
\begin{align}
\label{CDF LL}
G(x)=\frac{e^{\gamma} x^{\upsilon}}{1 + e^{\gamma} x^{\upsilon}}
\end{align}
and 
\begin{align}
\label{PDF LL}
f(x; \gamma,\upsilon) = \frac{e^{\gamma} \, \upsilon \, x^{\upsilon - 1}}{\left(1 + e^{\gamma} x^{\upsilon} \right)^2}; x>0,~\gamma,~\upsilon>0.
\end{align}

By using a transformation defined by \cite{balakrishnan2021record} on the CDF \eqref{CDF LL}, we obtain the record-based transmuted log-logistic (RBTLL) distribution, with the CDF and PDF given as follows:
\begin{align}
\label{CDF RBTLL}
F(x; \gamma, \upsilon,p)=\frac{e^{\gamma} x^{\upsilon}}{1 + e^{\gamma} x^{\upsilon}} + p \left(1 - \frac{e^{\gamma} x^{\upsilon}}{1 + e^{\gamma} x^{\upsilon}}\right) \log \left(1 - \frac{e^{\gamma} x^{\upsilon}}{1 + e^{\gamma} x^{\upsilon}}\right),
\end{align}

and
\begin{align}
\label{PDF RBTLL}
f(x; \gamma, \upsilon,p)=\frac{e^{\gamma} \upsilon x^{\upsilon-1}}{(1+e^{\gamma} x^{\upsilon})^2} \left(1 - p \left(\log\left(1 - \frac{e^{\gamma} x^{\upsilon}}{1+e^{\gamma} x^{\upsilon}}\right) + 1\right) \right),
\end{align}
where $~x>0,~\gamma,~\upsilon,~0<p<1$.

Figure \ref{pdf} shows possible shapes of PDF for the selected parameters.  

\begin{figure} [H]
    \centering
    \includegraphics[width=0.8\linewidth]{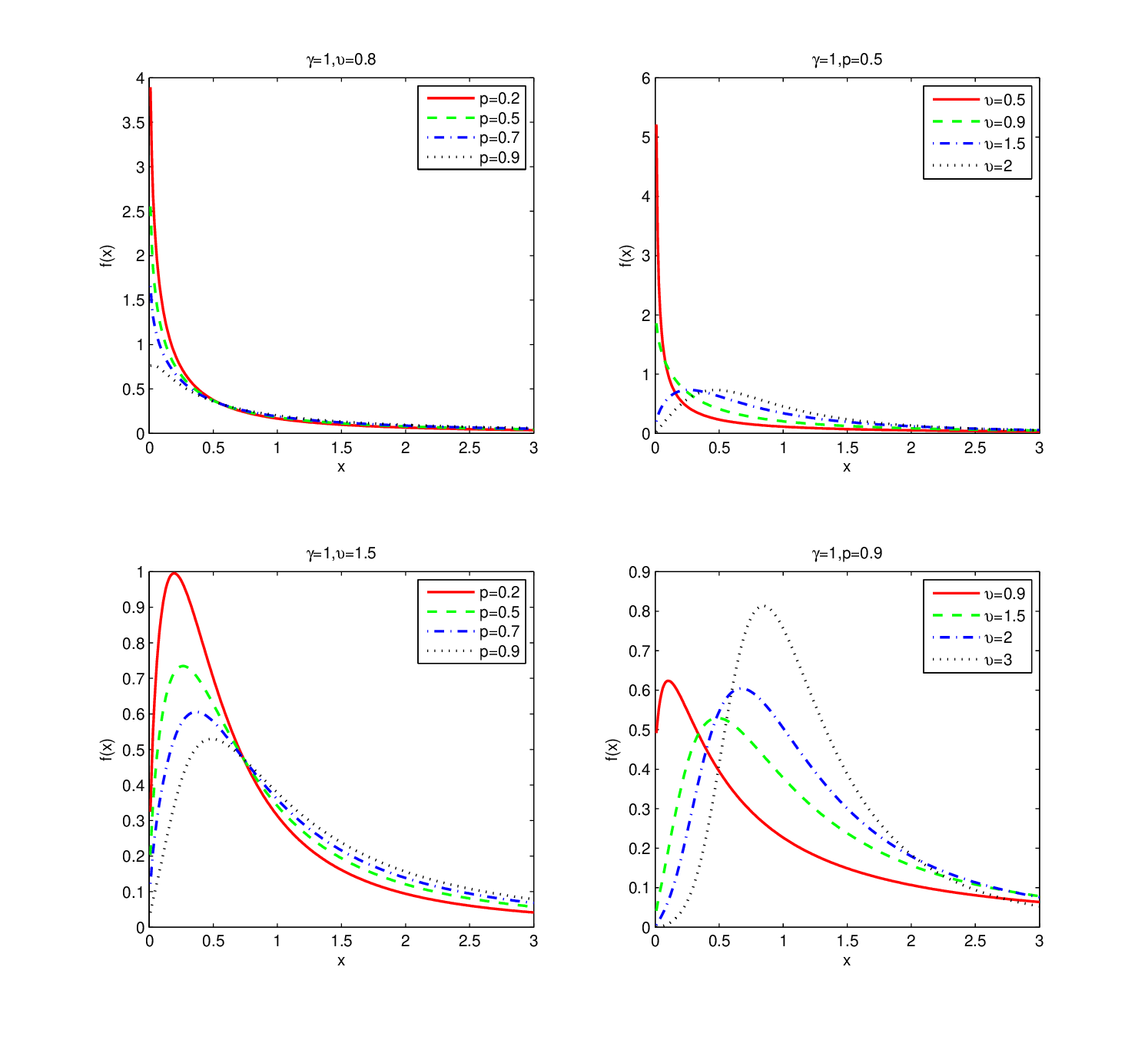}
    \caption{The PDFs of RBTLL distribution for the selected parameter values}
\label{pdf}
\end{figure}

\section{Distributional properties of the RBTLL distribution}
\label{Sec: propRBTLL}

In this section, we discuss some statistical properties of the RBTLL distribution such as the quantile function, hazard function (hf), moments, order statistics, and stochastic ordering.
\subsection{Quantile function}
The quantile function of RBTLL distribution is given as follows:
\begin{align}
\label{QF RBTLL}
 Q(x) =\exp\left( \frac{ \log\big(1 - u + p\,\mathrm{LambertW}\big( \frac{-1 + u}{p} e^{-1/p} \big) \big) - \log(-1 + u) - \gamma }{\upsilon} \right),
\end{align}
where, $0<u<1$ and \( W(\cdot) \) denotes the Lambert W function is defined as the inverse relation of the function \( w \mapsto w e^w \), that is: $W(z) \cdot e^{W(z)} = z$.

 \subsection{Hazard function}
The hf of RBTLL distribution is
\begin{equation} 
\label{hf}
 h(x;\gamma,\upsilon,p)=\frac{ e^{\gamma} \upsilon x^{\upsilon - 1} \left( 1 + p \log\left(1 + e^{\gamma} x^{\upsilon}\right) - p \right) }
     { \left(1 + e^{\gamma} x^{\upsilon}\right) \left( 1 + p \log\left(1 + e^{\gamma} x^{\upsilon}\right) \right). }
     \end{equation}

\subsubsection{Hazard Shape}
We examine the possible shapes of the hf in detail using Glaser's theorem \cite{glaser1980bathtub}. In this regard, we use the logarithmic derivative \( \psi(x) = \frac{d}{dx} \log h(x) = \frac{h'(x)}{h(x)} \).
Let \( z(x) = e^{\gamma} x^{\upsilon} \). Then the hf can be rewritten as follows
\[
h(x) = \frac{e^{\gamma} \upsilon x^{\upsilon - 1} \left(1 + p \log(1 + z) - p\right)}{(1 + z)\left(1 + p \log(1 + z)\right)}.
\]

We define
\[
\psi(x) = \frac{d}{dx} \log h(x) = \frac{d}{dx} \left( \log N(x) - \log D(x) \right) = \frac{N'(x)}{N(x)} - \frac{D'(x)}{D(x)},
\]
where
\[
N(x) = e^{\gamma} \upsilon x^{\upsilon - 1} \left(1 + p \log(1 + z) - p \right), \quad D(x) = (1 + z)\left(1 + p \log(1 + z)\right).
\]

We now analyze the sign of \( \psi(x) \) under three parameter regimes:

\textbf{Theorem.} Let \( h(x;\gamma,\upsilon,p) \) be as defined Eq. in (\ref{hf}). Then:
\begin{itemize}
    \item[(i)] If \( \upsilon < 1 \), then \( h(x) \) is decreasing.
    \item[(ii)] If \( \upsilon > 1 \), then \( h(x) \) is increasing.
    \item[(iii)] If \( \upsilon = 1 \), then \( h(x) \) may be unimodal or bathtub-shaped depending on \( \gamma \) and \( p \).
\end{itemize}

\begin{proof}
We obtain $\psi(x)$ as follows
\[
\psi(x) = \frac{\upsilon - 1}{x} - \frac{z'(x)}{1 + z} \cdot \left[ 1 + \frac{p}{1 + p \log(1 + z)} - \frac{p}{1 + p \log(1 + z) - p} \right].
\]

We assess the behavior of hf according to the sign of $\psi(x)$:

\begin{itemize}
    \item[(i)] For \( \upsilon < 1 \), the term \( \frac{\upsilon - 1}{x} < 0 \), and the second term is always positive . For this reason, \( \psi(x) < 0 \) for all \( x \), and the hf is decreasing.

    \item[(ii)] For \( \upsilon > 1 \), the first term \( \frac{\upsilon - 1}{x} > 0 \), and although the second term subtracts a positive quantity, the net effect for small and moderate \( x \) is dominated by the positive first term. As \( x \to \infty \), both terms decay, but the total expression remains non-negative. Therefore, \( \psi(x) > 0 \), indicating an increasing hf.

    \item[(iii)] When \( \upsilon = 1 \), the first term vanishes: \( \frac{\upsilon - 1}{x} = 0 \). Then the sign of \( \psi(x) \) is entirely governed by the remaining bracketed difference:
    \[
    \psi(x) = - \frac{z'(x)}{1 + z} \cdot \left[ 1 + \frac{p}{1 + p \log(1 + z)} - \frac{p}{1 + p \log(1 + z) - p} \right].
    \]
    For small \( x \), \( z \to 0 \), so the bracketed term is positive, and thus \( \psi(x) < 0 \). For large \( x \), the denominators behave such that the entire bracket may become negative, resulting in \( \psi(x) > 0 \). Therefore, \( \psi(x) \) changes sign, and the hazard function can exhibit a bathtub or unimodal shape depending on \( \gamma \) and \( p \).
\end{itemize}
Thus, the proof is completed.
\end{proof}

Figure \ref{hazard} illustrates possible shapes of hf for some parameters values.

\begin{figure} [H]
    \centering
    \includegraphics[width=0.8\linewidth]{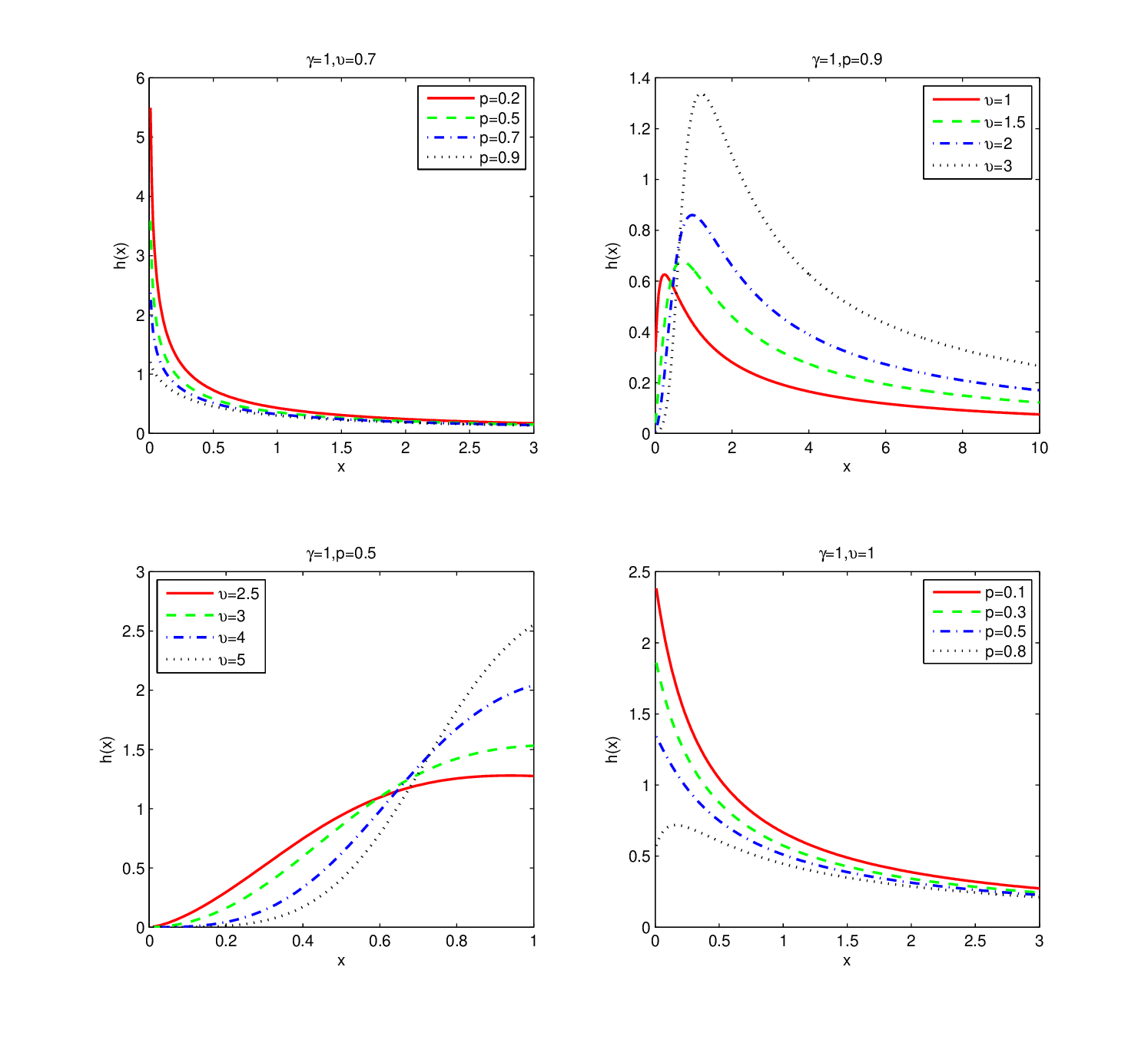}
    \caption{The HFs of RBTLL distribution for the selected parameter values}
\label{hazard}
\end{figure}

From Figure \ref{hazard}, we observe that the hf of the RBTLL distribution can be shaped as an increasing, decreasing, and unimodal. 

\subsection{Moments}

Let \(X\) be a distributed positive random variable RBTLL distribution with $\gamma$, $\upsilon$ and $p$ parameters whose \(r\) th moment (\(r \in \mathbb{N}\)) is given by

\[
E\left(X^{r}\right) = \int_0^\infty \frac{x^{r} e^{\gamma} \upsilon x^{\upsilon -1}}{\left(1 + e^{\gamma} x^{\upsilon}\right)^2} \left[ 1 + p \left( \log\left(1 + e^{\gamma} x^{\upsilon}\right) - 1 \right) \right] dx,
\quad \gamma, \upsilon > 0, \quad 0 < p < 1.
\]

Due to the complexity of the integrand, in particular the logarithmic term inside the integral, a closed-form expression for \(E(X^{r})\) is not available. However, the integral can be decomposed into two parts:

\[
E(X^{r}) = e^{\gamma} \upsilon \sum_{k=0}^\infty (-1)^k (k+1)(k+2) e^{\gamma k} \Gamma\left(\frac{r + \upsilon (k+1)}{\upsilon}\right) + \mathcal{L}_{p,\gamma,\upsilon}(r),
\]

where

\[
\mathcal{L}_{p,\gamma,\upsilon}(r) := \int_0^\infty \frac{x^{r} e^{\gamma} \upsilon x^{\upsilon -1}}{\left(1 + e^{\gamma} x^{\upsilon}\right)^2} \cdot p \left( \log\left(1 + e^{\gamma} x^{\upsilon}\right) - 1 \right) dx.
\]

 \(\mathcal{L}_{p,\gamma,\upsilon}(r)\), includes the logarithmic term and cannot be expressed in closed form or as a convergent power series due to its non-polynomial nature and problems of divergence at infinity. Therefore, it is calculated numerically.

\subsection{Order Statistics}
This section presents the PDFs of the first, \( r \)th, and \( n \)th order statistics for the RBTLL distribution.

Let \( X_{1:n} \leq X_{2:n} \leq \cdots \leq X_{n:n} \) refer to the order statistics of a random sample from the RBTLL distribution with the CDF in Eq. (\ref{CDF RBTLL}) and PDF in Eq. (\ref{PDF RBTLL}).

The probability density function of the minimum order statistic \( X_{1:n} \), \( X_{r:n} \), and \( X_{n:n} \) are given by respectively,
\begin{align}
f_{X_{1:n}}(x) =n \left( \frac{1}{1 + e^{\gamma} x^{\upsilon}} \right)^{n - 1} \cdot \frac{\gamma \upsilon x^{\gamma - 1}}{\left(1 + \left( \frac{x}{\upsilon} \right)^{\gamma} \right)^2},
\end{align}

\begin{align}
f_{X_{r:n}}(x) = \frac{n!}{(r - 1)! (n - r)!} \left( \frac{e^{\gamma} x^{\upsilon}}{1 + e^{\gamma} x^{\upsilon}} \right)^{r - 1} \left( \frac{1}{1 + e^{\gamma} x^{\upsilon}} \right)^{n - r} \cdot \frac{\gamma \upsilon x^{\gamma - 1}}{\left(1 + \left( \frac{x}{\upsilon} \right)^{\gamma} \right)^2},
\end{align}
and

\begin{align}
f_{X_{n:n}}(x) = n \left( \frac{e^{\gamma} x^{\upsilon}}{1 + e^{\gamma} x^{\upsilon}} \right)^{n - 1} \cdot \frac{\gamma \upsilon x^{\gamma - 1}}{\left(1 + \left( \frac{x}{\upsilon} \right)^{\gamma} \right)^2}.
\end{align}

\subsection{Stochastic Ordering}
In this subsection, we examine the stochastic ordering for the RBTLL distribution. In this regard, we define the following theorem.
\begin{theorem}
\label{lrordering} Let $X\sim RBTLL\left( \gamma ,\upsilon ,p_{1}\right) $ and
$Y\sim RBTLL\left( \gamma ,\upsilon ,p_{2}\right) .$ If $p_{1}<p_{2}$ then $X$
is less than $Y$ in the likelihood ratio order, that is, the ratio function
of the corresponding PDFs decreases in $x$.
\end{theorem}

\begin{proof}
The ratio of densities for any $x>0$ is as follows:
\begin{equation*}
g\left( x\right) =\frac{1 + p_1 \log\left(1 + e^{\gamma} x^{\upsilon}\right) - p_1}{1 + p_2 \log\left(1 + e^{\gamma} x^{\upsilon}\right) - p_2}.
\end{equation*}

Then, consider the derivative of $\log \left( g\left( x\right) \right) $ in $x$%
\begin{equation*}
\frac{d\log \left( g\left( x\right) \right) }{dx}=\frac{ e^{\alpha} x^{\beta} \beta (p_1 - p_2) }
     { x \left(1 + e^{\alpha} x^{\beta}\right) \left(1 + p_2 \log\left(1 + e^{\alpha} x^{\beta}\right) - p_2\right) \left(1 + p_1 \log\left(1 + e^{\alpha} x^{\beta}\right) - p_1\right) }<0
\end{equation*}%
for $p_{1}<p_{2}$, and thus the proof is completed.
\end{proof}

We notice that  it follows that $X$ is also less than $Y$ in the hazard ratio, the mean residual life, and the stochastic orders under the conditions
given in Theorem \ref{lrordering} from \cite{shaked1994stochastic}.

\section{Parameter Estimation}\label{Estimation methods}
In this section, we discuss point estimation for the RBTLL distribution. To estimate the parameters of RBTLL distribution, we utilize seven different methods. 

\subsection{Method of maximum likelihood}
In this subsection, we propose ML estimators (MLEs) of the parameters $\gamma$, $\upsilon$, and $p$ of the RBTLL distribution. 

Let ${X_1},{X_2}...{X_n}$ be a random sample from the RBTLL and ${x_1},{x_2}...{x_n}$  refer to the observed values of the sample. Then, the corresponding log-likelihood function is given by

\begin{align}
\label{mle}
\ell(\gamma, \upsilon, p) = \sum_{i=1}^n \Bigg[
\gamma + \log(\upsilon) + (\upsilon - 1) \log x_i - 2 \log \left(1 + e^{\gamma} x_i^{\upsilon} \right) + \log \left(1 - p \left(\log \left(1 - \frac{e^{\gamma} x_i^{\upsilon}}{1 + e^{\gamma} x_i^{\upsilon}}\right) + 1\right)\right)
\Bigg].
\end{align}

We obtain the MLEs of the parameters $\gamma, \upsilon$ and $p$ of the RBTLL by maximizing Eq. (\ref{mle}). This optimization problem is solved via numerical methods such as Newton-Raphson and Nelder Mead. 
\subsection{Method of least squares}
This subsection provides the LS estimators (LSEs) the $\gamma, \upsilon$ and $p$ parameters of the RBTLL distribution. This estimator is proposed by \cite{swain1988least} as an alternative to the MLE. The LSEs can be derived by minimizing the function given in Eq. (\ref{lse}).

\begin{align}
\label{lse}
LS(x_{i})&=\sum_{i=1}^{n}\left[F(x_{i:n})-\frac{i}{n+1}\right]^2,
\end{align}
where ${x_{i:n}}$ for $i = 1,2...n$ refer to the order statistics.
\subsection{Method of weighted least squares}
This subsection introduces the WLS estimators (WLSEs) of $\gamma, \upsilon$ and $p$ parameters for the RBTLL via the method proposed by \cite{swain1988least}. We obtain the WLSEs by minimizing the Eq. (\ref{wlse}).
\begin{align}
\label{wlse}
WLS(x_{i})&=\sum_{i=1}^{n}\frac{(n+1)^2(n+2)}{i(n-i+1)}\left[F(x_{i:n})-\frac{i}{n+1}\right]^2.
\end{align} 

\subsection{Method of Anderson-Darling}
In this subsection, we discuss the AD estimators (ADEs) of $\gamma, \upsilon$ and $p$ parameters. This method is related to the AD goodness-of-fit statistic proposed by \cite{anderson1952asymptotic}. The ADEs ${\hat \gamma _{ADE}},{\hat \upsilon _{ADE}}$, ${\hat p _{ADE}}$ of the parameters $\gamma, \upsilon$ and $p$ can be obtained by minimizing Eq (\ref{ade}).

\begin{align}
\label{ade}
AD(x_{i})&=-n-\frac{1}{n}\sum_{i=1}^{n}(2i-1)\left[\log F(x_{i:n})+\log S(x_{n-i-1:n})\right].
\end{align} 


\subsection{Method of Cram\'{e}r-von Mises}
In this section, we provide the CvM estimators (CvMEs) of parameters of RBTLL distribution. The Cram\'{e}r-von Mises method depends on minimizing the difference between the CDF and the empirical distribution functions, as proposed by \cite{choi1968estimation}. The CvMEs are computed by maximizing the function in Eq. (\ref{cvme}).

\begin{align}
\label{cvme}
C(x_{i})&=\frac{1}{12n}+\sum_{i=1}^{n}\left[F(x_{i:n})-\frac{2i-1}{2n}\right]^2.
\end{align}
\subsection{Method of maximum product of spacings}
In this subsection, we deal with the point estimation for the RBTLL distribution via the MPS method. The MPS method introduced by \cite{cheng1983estimating} and \cite{ranneby1984maximum} as an alternative to the ML method, the MPS method based on maximizing the function in Eq. (\ref{mpse}) to derive the MPS estimators (MPSEs). 
\begin{align}
\label{mpse}
\delta\left(x_{i} \right) =\frac{1}{n+1}\sum_{i=1}^{n+1}\log I_{i}(x_{i}),
\end{align}

where $I_{i}(x_{i})=F(x_{i:n})-F(x_{i-1:n})$, $F(x_{0:n})=0$ and $F(x_{n+1:n})=1.$

\subsection{Right tail Anderson Darling estimation method}
This subsection proposes the RTAD estimators (RTADEs) of $\gamma$, $\upsilon$, and $p$ parameters. The RTADEs are obtained by minimizing the Eq. (\ref{RTADE}).
\begin{equation} \label{RTADE}
    L\left( {\omega,\kappa,p} \right) = \frac{n}{2} - 2\sum\limits_{i = 1}^n {F\left( {{x_{i:n}}} \right) - \frac{1}{n}} \sum\limits_{i = 1}^n {\left( {2i - 1} \right)\log \left( {1 - F\left( {{x_{n - i + 1}}} \right)} \right)} .
\end{equation}

\section{Simulation}
\label{Sec: simulation}

This section provides a comprehensive Monte-Carlo (MC) Simulation study to assess the performance of the examined estimators in Section  \ref{Estimation methods}. We consider the parameter settings and initial values in the MC simulations as follows:
In all MC simulations, the sample sizes, $n=50,100,200,500$ with 5000 repetitions and the initial values of the $\gamma$, $\upsilon$ and $p$ parameters are considered as follows: \newline
$Case_{I} = \left( {\gamma  = 3, \upsilon  = 1.5,p = 0.75} \right)$,\\
$Case_{II} = \left( {\gamma  = 2,\upsilon  = 0.9,p = 0.5} \right)$,\\
$Case_{III} = \left( {\gamma  = 1.5,\upsilon  = 2,p = 0.4} \right)$,\\
$Case_{IV} = \left( {\gamma  = 2.5,\upsilon = 0.6,p = 0.3} \right)$,\\

We evaluate the performances of the mentioned estimators via the bias, mean squared error (MSE), and mean relative error (MRE) values. The corresponding formulas of these measures are 

$$Bias=\frac{1}{5000}\sum\limits_{i=1}^{5000}\left( \hat{\Theta}-\Theta
\right),$$

$$MSE=\frac{1}{5000}\sum\limits_{i=1}^{5000}\left( \hat{\Theta%
}-\Theta \right) ^{2},$$

$$MRE=\frac{1}{5000}\sum_{i=1}^{5000}\frac{|\hat{\Theta}-\Theta|}{\Theta},$$
where $\Theta=\left(\gamma,\upsilon,p\right)$. 
\subsection{Random sample generation}
In this subsection, we generate the random samples from the RBTLL $(\gamma, \upsilon, p)$ distribution. Therefore, we suggest an acceptance-rejection (AR) sampling algorithm.
We prefer the Weibull distribution which is a well-known distribution as the proposal distribution in the AR algorithm. The AR algorithm is given as follows:

\textbf{Algorithm 1.}

\textbf{A1.} Generate data on random variable $Y \sim Weibull(\vartheta,\varpi)$ 
with the PDF $g$ given as follows: 
\begin{equation*}
g\left( \vartheta,\varpi \right)=\vartheta \varpi x^{\varpi -1}e^{-\vartheta
x^\varpi}.
\end{equation*}

\textbf{A2.} Generate $U$ from standard uniform distribution(independent of $%
Y$).

\textbf{A3.} If%
\begin{equation*}
U<\frac{f\left( Y;\gamma,\upsilon,p\right) }{k\times g\left(
Y;\vartheta,\varpi \right) }
\end{equation*}%
then set $X=Y$ (\textquotedblleft accept\textquotedblright ); otherwise go
back to A1 (\textquotedblleft reject\textquotedblright ), where the PDF $f$ $%
\left( .\right) $ is given as in Eq. (\ref{PDF RBTLL})  and 
\begin{equation*}
k=\underset{z\in 
\mathbb{R}
_{+}}{\max }\frac{f\left( z;\gamma,\upsilon,p\right) }{g\left(
z;\vartheta,\varpi \right) }.
\end{equation*}%

We generate random data on $X$ from the $RBTLL(\gamma, \upsilon, p)$ via AR algorithm. We utilize Algorithm 1 in all MC simulations.

The results of MC simulations are given in Tables \ref{sim:tab1}-\ref{sim:tab4}.


\begin{table}[H]
\centering
\caption{The biases, MSEs and MREs for $\gamma=3$, $\upsilon=1.5$ and $p=0.75$}
\scalebox{1} {\ \label{sim:tab1}
\begin{tabular}{ccccccccccc}\hline
          &      &         & Bias    &         &        & MSE    &        &        & MRE    &        \\\hline
Estimator & $n$    & $\hat\gamma$   & $\hat\upsilon$   & $\hat p$  & $\hat\gamma$  & $\hat\upsilon$  & $\hat p$ & $\hat\gamma$  & $\hat\upsilon$  & $\hat p$ \\\hline

MLE       & 50  & -0.1591 & 0.0357  & -0.1255 & 0.2974 & 0.0404 & 0.0887 & 0.1402 & 0.1035 & 0.2932 \\
          & 100 & -0.1507 & 0.0252  & -0.1117 & 0.2287 & 0.0208 & 0.0856 & 0.1192 & 0.0740 & 0.2963 \\
          & 200 & -0.1435 & 0.0081  & -0.0963 & 0.1972 & 0.0109 & 0.0793 & 0.1050 & 0.0558 & 0.2757 \\
          & 500 & -0.1158 & 0.0061  & -0.0801 & 0.1468 & 0.0052 & 0.0643 & 0.0843 & 0.0390 & 0.2346 \\\hline
LSE       & 50  & -0.2800 & -0.0529 & -0.1313 & 0.7102 & 0.0531 & 0.2553 & 0.2112 & 0.1226 & 0.5238 \\
          & 100 & -0.2487 & -0.0376 & -0.1234 & 0.5391 & 0.0269 & 0.1978 & 0.1799 & 0.0885 & 0.4572 \\
          & 200 & -0.2360 & -0.0260 & -0.1273 & 0.4025 & 0.0146 & 0.1552 & 0.1502 & 0.0649 & 0.3933 \\
          & 500 & -0.1535 & -0.0116 & -0.0888 & 0.2309 & 0.0073 & 0.0975 & 0.1061 & 0.0457 & 0.2977 \\\hline
WLSE      & 50  & -0.2678 & -0.0318 & -0.1422 & 0.5411 & 0.0420 & 0.1811 & 0.1852 & 0.1080 & 0.4139 \\
          & 100 & -0.2201 & -0.0146 & -0.1248 & 0.3846 & 0.0212 & 0.1380 & 0.1495 & 0.0770 & 0.3675 \\
          & 200 & -0.1926 & -0.0093 & -0.1138 & 0.2828 & 0.0113 & 0.1084 & 0.1239 & 0.0570 & 0.3170 \\
          & 500 & -0.1237 & -0.0026 & -0.0779 & 0.1659 & 0.0054 & 0.0711 & 0.0884 & 0.0395 & 0.2448 \\\hline
ADE       & 50  & -0.2144 & -0.0021 & -0.1325 & 0.4508 & 0.0407 & 0.1470 & 0.1687 & 0.1051 & 0.3651 \\
          & 100 & -0.1920 & -0.0022 & -0.1171 & 0.3549 & 0.0208 & 0.1253 & 0.1434 & 0.0756 & 0.3464 \\
          & 200 & -0.1853 & -0.0048 & -0.1128 & 0.2777 & 0.0111 & 0.1066 & 0.1226 & 0.0566 & 0.3126 \\
          & 500 & -0.1237 & -0.0017 & -0.0785 & 0.1698 & 0.0054 & 0.0725 & 0.0892 & 0.0396 & 0.2461 \\\hline
CvME      & 50  & -0.2278 & -0.0128 & -0.1295 & 0.7283 & 0.0551 & 0.2734 & 0.2175 & 0.1223 & 0.5511 \\
          & 100 & -0.2210 & -0.0175 & -0.1214 & 0.5435 & 0.0271 & 0.2047 & 0.1833 & 0.0881 & 0.4691 \\
          & 200 & -0.2224 & -0.0157 & -0.1266 & 0.4025 & 0.0146 & 0.1580 & 0.1515 & 0.0647 & 0.3992 \\
          & 500 & -0.1478 & -0.0074 & -0.0885 & 0.2309 & 0.0074 & 0.0983 & 0.1069 & 0.0459 & 0.3001 \\\hline
MPSE      & 50  & -0.2670 & -0.0749 & -0.1016 & 0.3480 & 0.0385 & 0.0809 & 0.1427 & 0.1067 & 0.2293 \\
          & 100 & -0.2025 & -0.0371 & -0.0921 & 0.2442 & 0.0189 & 0.0771 & 0.1161 & 0.0723 & 0.2476 \\
          & 200 & -0.1591 & -0.0269 & -0.0767 & 0.2022 & 0.0102 & 0.0705 & 0.0989 & 0.0543 & 0.2345 \\
          & 500 & -0.1137 & -0.0111 & -0.0641 & 0.1405 & 0.0047 & 0.0581 & 0.0783 & 0.0372 & 0.2111 \\\hline
TADE      & 50  & -0.1467 & -0.0581 & -0.0272 & 0.6865 & 0.0511 & 0.2577 & 0.2091 & 0.1205 & 0.5576 \\
          & 100 & -0.2039 & -0.0437 & -0.0835 & 0.5977 & 0.0284 & 0.2107 & 0.1898 & 0.0899 & 0.4848 \\
          & 200 & -0.2245 & -0.0305 & -0.1131 & 0.4500 & 0.0146 & 0.1665 & 0.1606 & 0.0641 & 0.4198 \\
          & 500 & -0.1530 & -0.0103 & -0.0887 & 0.2454 & 0.0064 & 0.0980 & 0.1123 & 0.0428 & 0.3043 

\\\hline
\end{tabular}
}
\end{table}

\begin{table}[H]
\centering
\caption{The biases, MSEs and MREs for $\gamma=2$, $\upsilon=0.9$ and $p=0.5$}
\scalebox{1} {\ \label{sim:tab2}
\begin{tabular}{ccccccccccc}\hline
          &      &         & Bias    &         &        & MSE    &        &        & MRE    &        \\\hline
Estimator & $n$    & $\hat\gamma$   & $\hat\upsilon$   & $\hat p$  & $\hat\gamma$  & $\hat\upsilon$  & $\hat p$ & $\hat\gamma$  & $\hat\upsilon$  & $\hat p$ \\\hline

MLE       & 50  & 0.2393  & 0.0726  & 0.0297  & 0.2551 & 0.0165 & 0.0517 & 0.1984 & 0.1093 & 0.3842 \\
          & 100 & 0.1984  & 0.0698  & 0.0119  & 0.1876 & 0.0106 & 0.0484 & 0.1705 & 0.0903 & 0.3690 \\
          & 200 & 0.1500  & 0.0603  & -0.0176 & 0.1441 & 0.0065 & 0.0437 & 0.1443 & 0.0734 & 0.3542 \\
          & 500 & -0.0181 & 0.0633  & -0.1195 & 0.0580 & 0.0051 & 0.0324 & 0.0907 & 0.0712 & 0.3217 \\\hline
LSE       & 50  & 0.2391  & -0.0383 & 0.1439  & 0.5715 & 0.0168 & 0.2156 & 0.3305 & 0.1176 & 0.8435 \\
          & 100 & 0.2543  & -0.0187 & 0.1376  & 0.4298 & 0.0088 & 0.1668 & 0.2924 & 0.0832 & 0.7409 \\
          & 200 & 0.2893  & -0.0135 & 0.1461  & 0.3393 & 0.0044 & 0.1235 & 0.2633 & 0.0584 & 0.6339 \\
          & 500 & 0.1950  & 0.0045  & 0.0731  & 0.2103 & 0.0020 & 0.0737 & 0.2015 & 0.0384 & 0.4731 \\\hline
WLSE      & 50  & 0.2213  & -0.0059 & 0.0974  & 0.4484 & 0.0124 & 0.1403 & 0.2870 & 0.0978 & 0.6684 \\
          & 100 & 0.1679  & 0.0164  & 0.0449  & 0.3422 & 0.0070 & 0.1083 & 0.2489 & 0.0720 & 0.5731 \\
          & 200 & 0.0905  & 0.0202  & -0.0164 & 0.2868 & 0.0040 & 0.0935 & 0.2257 & 0.0546 & 0.5240 \\
          & 500 & -0.1985 & 0.0271  & -0.2002 & 0.2417 & 0.0031 & 0.0992 & 0.1942 & 0.0523 & 0.5265 \\\hline
ADE       & 50  & 0.2417  & 0.0086  & 0.0952  & 0.4505 & 0.0131 & 0.1375 & 0.2862 & 0.1003 & 0.6622 \\
          & 100 & 0.2273  & 0.0193  & 0.0806  & 0.3343 & 0.0073 & 0.1069 & 0.2494 & 0.0737 & 0.5754 \\
          & 200 & 0.2270  & 0.0174  & 0.0746  & 0.2609 & 0.0039 & 0.0839 & 0.2190 & 0.0544 & 0.5052 \\
          & 500 & 0.0711  & 0.0292  & -0.0296 & 0.1520 & 0.0024 & 0.0534 & 0.1605 & 0.0452 & 0.3941 \\\hline
CvME      & 50  & 0.2766  & -0.0149 & 0.1462  & 0.6343 & 0.0171 & 0.2341 & 0.3496 & 0.1161 & 0.8784 \\
          & 100 & 0.2738  & -0.0067 & 0.1390  & 0.4567 & 0.0089 & 0.1746 & 0.3020 & 0.0828 & 0.7580 \\
          & 200 & 0.2989  & -0.0073 & 0.1465  & 0.3511 & 0.0044 & 0.1264 & 0.2676 & 0.0582 & 0.6404 \\
          & 500 & 0.1987  & 0.0071  & 0.0730  & 0.2132 & 0.0020 & 0.0743 & 0.2026 & 0.0391 & 0.4750 \\\hline
MPSE      & 50  & 0.0870  & -0.0095 & 0.0131  & 0.3692 & 0.0115 & 0.1174 & 0.2373 & 0.0945 & 0.5622 \\
          & 100 & 0.0855  & 0.0216  & -0.0140 & 0.2894 & 0.0068 & 0.0918 & 0.2132 & 0.0704 & 0.4980 \\
          & 200 & 0.0392  & 0.0253  & -0.0555 & 0.2848 & 0.0052 & 0.0958 & 0.2070 & 0.0628 & 0.4991 \\
          & 500 & -0.1790 & 0.0326  & -0.1953 & 0.2642 & 0.0052 & 0.1055 & 0.1762 & 0.0690 & 0.4993 \\\hline
TADE      & 50  & -0.1191 & -0.0469 & -0.0871 & 1.0312 & 0.0230 & 0.3405 & 0.4494 & 0.1356 & 1.0502 \\
          & 100 & -0.1480 & -0.0269 & -0.1234 & 0.7695 & 0.0141 & 0.2554 & 0.3747 & 0.1044 & 0.8764 \\
          & 200 & -0.2775 & -0.0309 & -0.2109 & 0.6830 & 0.0102 & 0.2317 & 0.3498 & 0.0884 & 0.8185 \\
          & 500 & -0.4909 & -0.0300 & -0.3422 & 0.6346 & 0.0077 & 0.2274 & 0.3294 & 0.0799 & 0.8080 

\\\hline
\end{tabular}
}
\end{table}

\begin{table}[H]
\centering
\caption{The biases, MSEs and MREs for $\gamma=1.5$, $\upsilon=2$ and $p=0.4$}
\scalebox{1} {\ \label{sim:tab3}
\begin{tabular}{ccccccccccc}\hline
          &      &         & Bias    &         &        & MSE    &        &        & MRE    &        \\\hline
Estimator & $n$    & $\hat\gamma$   & $\hat\upsilon$   & $\hat p$  & $\hat\gamma$  & $\hat\upsilon$  & $\hat p$ & $\hat\gamma$  & $\hat\upsilon$  & $\hat p$ \\\hline

 MLE       & 50  & 0.2360 & -0.0082 & 0.1399 & 0.2872 & 0.0633 & 0.0997 & 0.2940 & 0.0981 & 0.6627 \\
          & 100 & 0.1976 & -0.0447 & 0.1388 & 0.2612 & 0.0336 & 0.1017 & 0.2880 & 0.0743 & 0.6811 \\
          & 200 & 0.1729 & -0.0501 & 0.1266 & 0.2525 & 0.0191 & 0.1046 & 0.2897 & 0.0555 & 0.6995 \\
          & 500 & 0.0806 & -0.0519 & 0.0656 & 0.2002 & 0.0105 & 0.0818 & 0.2590 & 0.0413 & 0.6181 \\\hline
LSE       & 50  & 0.1169 & -0.1421 & 0.1028 & 0.6294 & 0.0952 & 0.2571 & 0.4683 & 0.1260 & 1.1648 \\
          & 100 & 0.1090 & -0.1285 & 0.1055 & 0.5175 & 0.0597 & 0.2058 & 0.4343 & 0.1005 & 1.0429 \\
          & 200 & 0.1378 & -0.0981 & 0.1175 & 0.3916 & 0.0321 & 0.1617 & 0.3768 & 0.0722 & 0.9169 \\
          & 500 & 0.0395 & -0.0789 & 0.0468 & 0.2762 & 0.0164 & 0.1099 & 0.3144 & 0.0512 & 0.7492 \\\hline
WLSE      & 50  & 0.1780 & -0.1116 & 0.1331 & 0.4763 & 0.0759 & 0.1872 & 0.4030 & 0.1120 & 0.9999 \\
          & 100 & 0.1385 & -0.1002 & 0.1162 & 0.3948 & 0.0430 & 0.1512 & 0.3741 & 0.0847 & 0.8918 \\
          & 200 & 0.1662 & -0.0746 & 0.1288 & 0.3050 & 0.0233 & 0.1243 & 0.3313 & 0.0616 & 0.8001 \\
          & 500 & 0.0666 & -0.0630 & 0.0598 & 0.2224 & 0.0118 & 0.0891 & 0.2820 & 0.0435 & 0.6702 \\\hline
ADE       & 50  & 0.1997 & -0.0706 & 0.1347 & 0.4238 & 0.0683 & 0.1608 & 0.3783 & 0.1046 & 0.9296 \\
          & 100 & 0.1579 & -0.0832 & 0.1239 & 0.3749 & 0.0399 & 0.1428 & 0.3650 & 0.0809 & 0.8703 \\
          & 200 & 0.1654 & -0.0696 & 0.1271 & 0.3037 & 0.0224 & 0.1240 & 0.3320 & 0.0600 & 0.8036 \\
          & 500 & 0.0663 & -0.0617 & 0.0592 & 0.2243 & 0.0117 & 0.0898 & 0.2835 & 0.0433 & 0.6744 \\\hline
CvME      & 50  & 0.1376 & -0.0904 & 0.1005 & 0.6800 & 0.0903 & 0.2745 & 0.4881 & 0.1202 & 1.1990 \\
          & 100 & 0.1219 & -0.1021 & 0.1060 & 0.5380 & 0.0558 & 0.2128 & 0.4429 & 0.0967 & 1.0560 \\
          & 200 & 0.1438 & -0.0846 & 0.1175 & 0.4000 & 0.0303 & 0.1646 & 0.3806 & 0.0699 & 0.9230 \\
          & 500 & 0.0412 & -0.0732 & 0.0463 & 0.2782 & 0.0157 & 0.1106 & 0.3152 & 0.0498 & 0.7504 \\\hline
MPSE      & 50  & 0.2014 & -0.1634 & 0.1690 & 0.2967 & 0.0815 & 0.1277 & 0.2994 & 0.1181 & 0.8050 \\
          & 100 & 0.1661 & -0.1361 & 0.1481 & 0.2810 & 0.0463 & 0.1215 & 0.3087 & 0.0894 & 0.8037 \\
          & 200 & 0.1545 & -0.1002 & 0.1311 & 0.2622 & 0.0252 & 0.1146 & 0.3059 & 0.0653 & 0.7725 \\
          & 500 & 0.0853 & -0.0787 & 0.0771 & 0.2212 & 0.0134 & 0.0918 & 0.2820 & 0.0479 & 0.6832 \\\hline
TADE      & 50  & 0.2496 & -0.1350 & 0.1983 & 0.7196 & 0.0955 & 0.3131 & 0.5044 & 0.1275 & 1.2715 \\
          & 100 & 0.1850 & -0.1254 & 0.1587 & 0.5587 & 0.0586 & 0.2295 & 0.4535 & 0.0986 & 1.1032 \\
          & 200 & 0.1646 & -0.0988 & 0.1366 & 0.4428 & 0.0341 & 0.1758 & 0.4051 & 0.0735 & 0.9597 \\
          & 500 & 0.0475 & -0.0776 & 0.0522 & 0.2897 & 0.0171 & 0.1134 & 0.3217 & 0.0514 & 0.7592 

\\\hline
\end{tabular}
}
\end{table}


\begin{table}[H]
\centering
\caption{The biases, MSEs and MREs for $\gamma=2.5$, $\upsilon=0.6$ and $p=0.3$}
\scalebox{1} {\ \label{sim:tab4}
\begin{tabular}{ccccccccccc}\hline
          &      &         & Bias    &         &        & MSE    &        &        & MRE    &        \\\hline
Estimator & $n$    & $\hat\gamma$   & $\hat\upsilon$   & $\hat p$  & $\hat\gamma$  & $\hat\upsilon$  & $\hat p$ & $\hat\gamma$  & $\hat\upsilon$  & $\hat p$ \\\hline

MLE       & 50  & 0.5881  & 0.0532  & 0.2189  & 0.5814 & 0.0078 & 0.1030 & 0.2502 & 0.1146 & 0.8349 \\
          & 100 & 0.5667  & 0.0457  & 0.2162  & 0.5077 & 0.0048 & 0.1023 & 0.2382 & 0.0922 & 0.8153 \\
          & 200 & 0.4927  & 0.0450  & 0.1683  & 0.3832 & 0.0034 & 0.0755 & 0.2053 & 0.0812 & 0.6695 \\
          & 500 & 0.3436  & 0.0456  & 0.0803  & 0.1827 & 0.0027 & 0.0260 & 0.1428 & 0.0770 & 0.3646 \\\hline
LSE       & 50  & 0.4889  & -0.0186 & 0.3247  & 0.8070 & 0.0081 & 0.3079 & 0.3139 & 0.1208 & 1.6646 \\
          & 100 & 0.5505  & -0.0113 & 0.3419  & 0.7036 & 0.0044 & 0.2675 & 0.2992 & 0.0876 & 1.5467 \\
          & 200 & 0.5766  & -0.0017 & 0.3322  & 0.6065 & 0.0022 & 0.2176 & 0.2795 & 0.0629 & 1.3923 \\
          & 500 & 0.5394  & 0.0083  & 0.2902  & 0.4421 & 0.0011 & 0.1455 & 0.2404 & 0.0431 & 1.1242 \\\hline
WLSE      & 50  & 0.5298  & 0.0019  & 0.3012  & 0.7072 & 0.0060 & 0.2200 & 0.2896 & 0.1013 & 1.3745 \\
          & 100 & 0.5130  & 0.0100  & 0.2648  & 0.6344 & 0.0033 & 0.1814 & 0.2774 & 0.0741 & 1.2377 \\
          & 200 & 0.4182  & 0.0200  & 0.1784  & 0.4794 & 0.0021 & 0.1235 & 0.2366 & 0.0619 & 0.9801 \\
          & 500 & 0.1631  & 0.0245  & 0.0128  & 0.2490 & 0.0018 & 0.0555 & 0.1684 & 0.0614 & 0.6207 \\\hline
ADE       & 50  & 0.5583  & 0.0128  & 0.2932  & 0.7313 & 0.0065 & 0.2093 & 0.2963 & 0.1060 & 1.3506 \\
          & 100 & 0.5723  & 0.0125  & 0.2980  & 0.6389 & 0.0035 & 0.1909 & 0.2795 & 0.0773 & 1.2660 \\
          & 200 & 0.5503  & 0.0180  & 0.2683  & 0.5233 & 0.0020 & 0.1509 & 0.2514 & 0.0608 & 1.1015 \\
          & 500 & 0.4338  & 0.0252  & 0.1843  & 0.3197 & 0.0014 & 0.0800 & 0.1956 & 0.0515 & 0.7687 \\\hline
CvME      & 50  & 0.5473  & -0.0028 & 0.3265  & 0.9185 & 0.0085 & 0.3269 & 0.3352 & 0.1207 & 1.7108 \\
          & 100 & 0.5825  & -0.0033 & 0.3443  & 0.7568 & 0.0045 & 0.2761 & 0.3102 & 0.0877 & 1.5685 \\
          & 200 & 0.5926  & 0.0025  & 0.3330  & 0.6311 & 0.0023 & 0.2208 & 0.2849 & 0.0639 & 1.3997 \\
          & 500 & 0.5458  & 0.0101  & 0.2902  & 0.4501 & 0.0011 & 0.1462 & 0.2425 & 0.0441 & 1.1244 \\\hline
MPSE      & 50  & 0.3090  & -0.0011 & 0.1664  & 0.5110 & 0.0052 & 0.1475 & 0.2323 & 0.0930 & 1.0717 \\
          & 100 & 0.3141  & 0.0102  & 0.1358  & 0.5329 & 0.0034 & 0.1361 & 0.2408 & 0.0761 & 1.0226 \\
          & 200 & 0.2929  & 0.0202  & 0.0963  & 0.4445 & 0.0026 & 0.1084 & 0.2221 & 0.0695 & 0.8789 \\
          & 500 & 0.1100  & 0.0223  & -0.0187 & 0.3232 & 0.0027 & 0.0741 & 0.1818 & 0.0742 & 0.6555 \\\hline
TADE      & 50  & 0.1919  & -0.0248 & 0.1394  & 1.2496 & 0.0097 & 0.3500 & 0.4003 & 0.1316 & 1.7919 \\
          & 100 & 0.1584  & -0.0187 & 0.0960  & 0.9912 & 0.0064 & 0.2615 & 0.3582 & 0.1065 & 1.5472 \\
          & 200 & 0.0574  & -0.0123 & 0.0141  & 0.7336 & 0.0044 & 0.1874 & 0.3068 & 0.0892 & 1.2797 \\
          & 500 & -0.1495 & -0.0125 & -0.1113 & 0.5162 & 0.0034 & 0.1272 & 0.2555 & 0.0801 & 1.0348 

\\\hline
\end{tabular}
}
\end{table}
From Tables \ref{sim:tab1}-\ref{sim:tab4}, we conclude the following inferences.
\begin{itemize}
    \item The results of MC simulation study show that the MSE,MRE, and bias values generally decrease as the sample size $n$ increases.
 \item It is seen that $\upsilon$ is the parameter with the smallest MSE value, while $\gamma$ is the parameter with the largest MSE value.
 
 \item In the estimation of the parameter $\gamma$, MLE is superior to the other six estimators according to the MSE criterion.

 \item In the estimation of the parameter $\upsilon$, MLE as well as MPSE performed well according to the MSE criterion at small sample sizes. We observe that MPSE sometimes has a smaller MSE than MLE at small sample sizes.

\item When the MSE values of the parameter $\upsilon$ are examined, it is determined that LSE and CvME are worthy competitors to MLE for large sample sizes. It should be noted that LSE and CvME have MSE values smaller than MLE in some parameter cases.

\item Similarly to the parameter $\gamma$, MLE is the best estimator according to the MSE criterion to estimate the parameter $p$.

\item Based on the findings of the simulation study, we finally recommend MLE, MPSE, LSE and CvME estimators to estimate the three parameters $\gamma,$ $\upsilon$ and $p$ of the RBTLL distribution among the analyzed seven estimators.
 
\end{itemize}

\section{Real data analysis}
\label{Sec: realdata}

We analyze two practical data sets to assess the superiority of the RBTLL distribution over its rivals such as log-logistic (LL), extended log-logistic (ELL) \cite{lima2017extended}, transmuted Weibull (TW) \cite{aryal2011transmuted}, Weibull (W) in modeling real-life data. For the comparison the fitted models, we consider some selection criteria as follows: Kolmogorov-Smirnov statistics (KS), Anderson-Darling statistics (AD), Cramer-von-Mises statistics (CvM), and their p-values.
\subsection{Reactor pump failure data}
In this subsection, we analyze real-world data provided by \cite{suprawhardana1999total}
on the time (thousands of hours) between failures of secondary reactor pumps. The reactor pump failure data are:
2.160, 0.746, 0.402, 0.954, 0.491, 6.560, 4.992, 0.347, 0.150, 0.358, 0.101, 1.359, 3.465,
1.060, 0.614, 1.921, 4.082, 0.199, 0.605, 0.273, 0.070, 0.062, 5.320.

Table \ref{tab:real1} shows the selection criteria as well as the MLEs and the corresponding standard errors (SEs) of the parameters for the fitted distributions. Figure \ref{fig:real1cdf} illustrates the fitted CDFs while Figure \ref{fig:real1pdf} shows the fitted PDFs for the reactor pump failure data. Figure \ref{Fig: nonplot1} illustrates the non-parametric plots for reactor pump failure data.

\begin{figure} [H]
    \centering
    \includegraphics[width=0.8\linewidth]{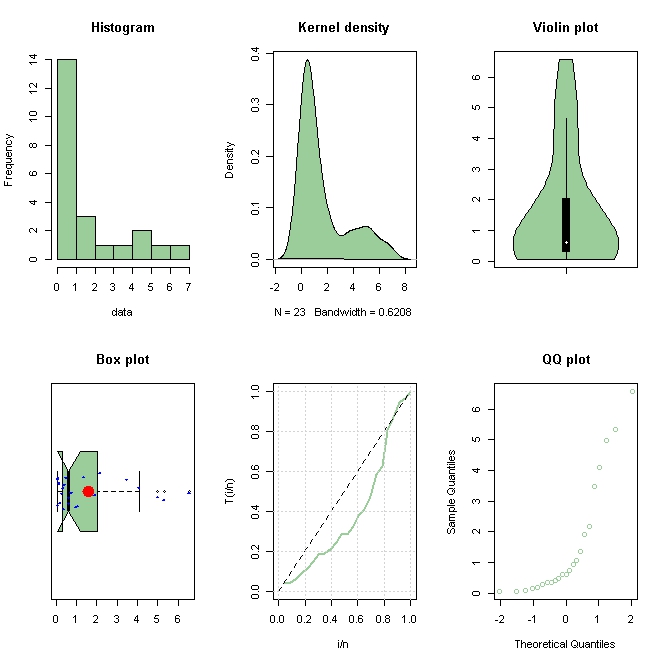}
    \caption{Non-parametric plots for reactor pump failure data}
    \label{Fig: nonplot1}
\end{figure}

\begin{table}[H]
\centering
\caption{The comparison statistics, MLEs and SEs of the parameters for the reactor pump failure data}%
\label{tab:real1}
\scalebox{0.6} {\ 
\begin{tabular}{cccccccccccccccc}
\hline
Model & $-2\log \ell$ & KS     & AD     & CvM    & p-value(KS) & p-value(AD) & p-value(CvM) & $\hat{\gamma}$    & $\hat{\upsilon}$  & $\hat{p}$       & $\hat{\lambda}$       & SE($\hat{\gamma}$) & SE($\hat{\upsilon}$) & SE($\hat{p}$) & SE($\hat{\lambda}$)  \\\hline
RBTLL & 65.2329 & 0.0910 & 0.2278 & 0.0245 & 0.9820      & 0.9809      & 0.9920       & 1.0253  & 1.3240  & 0.3483   & -                          & 1.7055    & 0.2457      & 0.9870   & -          \\
LL    & 65.2273 & 0.0933 & 0.2321 & 0.0259 & 0.9770      & 0.9790      & 0.9893       & 0.7067  & 1.2294  & -        & -                          & 0.2122    & 0.2096      & -        & -          \\
ELL   & 63.7915 & 0.0968 & 0.2284 & 0.0256 & 0.9680      & 0.9806      & 0.9898       & 15.7819 & 7.5304  & 102.2741 & 0.3607                     & 31.0218   & 19.1216     & 683.6523 & 0.2998     \\
TW    & 64.7601 & 0.1115 & 0.3718 & 0.0541 & 0.9075      & 0.8750      & 0.8561       & 0.8535  & 1.7451  & 0.3311   & -                          & 0.1430    & 0.8440      & 0.5849   & -          \\
W     & 65.0278 & 0.1184 & 0.4041 & 0.0619 & 0.8667      & 0.8432      & 0.8072       & 0.8077  & 1.3915  & -        & -                          & 0.1298    & 0.3805      & -        & -    \\\hline 
\end{tabular}
}
\end{table}

\begin{figure}[H]
\centerline{\includegraphics[width=4.96in,height=2.82in]{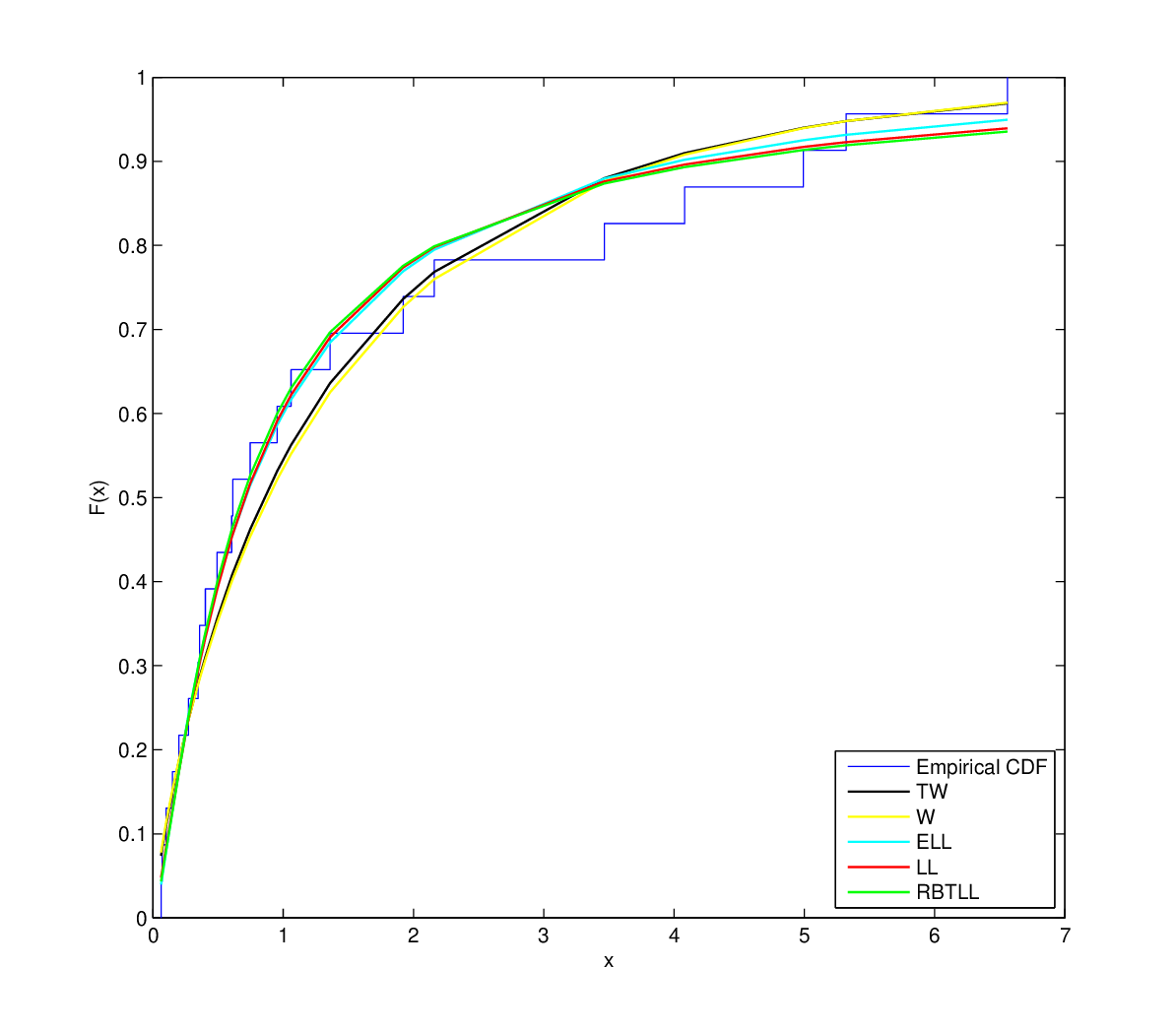}}
\caption{The fitted CDFs for the reactor pump failure data set}
\label{fig:real1cdf}
\end{figure}

\begin{figure}[H]
\centerline{\includegraphics[width=4.89in,height=2.78in]{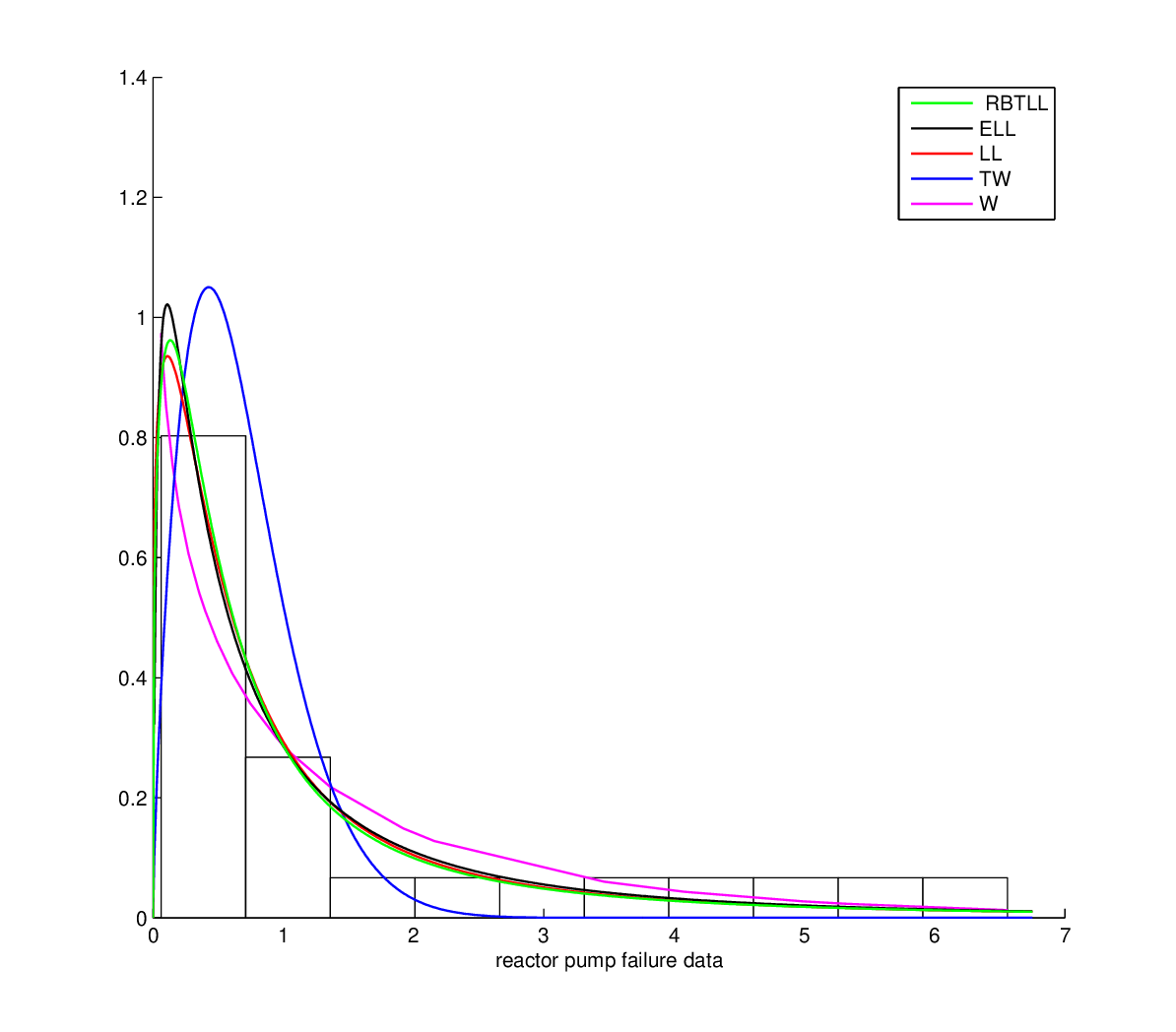}}
\caption{The fitted PDFs for the reactor pump failure data set}
\label{fig:real1pdf}
\end{figure}

\subsection{Petroleum rock data}
This subsection presents second practical data example to compare the fits of the RBTLL distribution and its rivals. For this reason, we consider the right skewed dataset discussed by \cite{moutinho2012beta} consists the 48 observations on petroleum rock samples. These data set relates to twelve core samples sampled from petroleum reservoirs in four cross-sections. Each core sample was measured for permeability and each cross-section has the some variables such as total pore area, total pore perimeter and shape. We consider the shape perimeter by squared (area) variable and the petroleum rock data are
0.0903296, 0.2036540, 0.2043140, 0.2808870, 0.1976530, 0.3286410, 0.1486220, 0.1623940, 0.2627270, 0.1794550, 0.3266350, 0.2300810, 0.1833120, 0.1509440, 0.2000710, 0.1918020, 0.1541920, 0.4641250, 0.1170630, 0.1481410, 0.1448100, 0.1330830, 0.2760160, 0.4204770, 0.1224170, 0.2285950, 0.1138520, 0.2252140, 0.1769690, 0.2007440, 0.1670450, 0.2316230, 0.2910290, 0.3412730, 0.4387120, 0.2626510, 0.1896510, 0.1725670, 0.2400770, 0.3116460, 0.1635860, 0.1824530, 0.1641270, 0.1534810, 0.1618650, 0.2760160, 0.2538320, 0.2004470.

Table \ref{tab:real2} shows the comparison statistics as weel as the MLEs and corresponding SEs of the parameters for the fitted distributions. Figure \ref{fig:real2cdf} illustrates the fitted CDFs while Figure \ref{fig:real2pdf} shows the fitted PDFs for petroleum rock data. Figure \ref{Fig: nonplot2} illustrates the non-parametric plots for petroleum rock data.

\begin{figure} [H]
    \centering
    \includegraphics[width=0.8\linewidth]{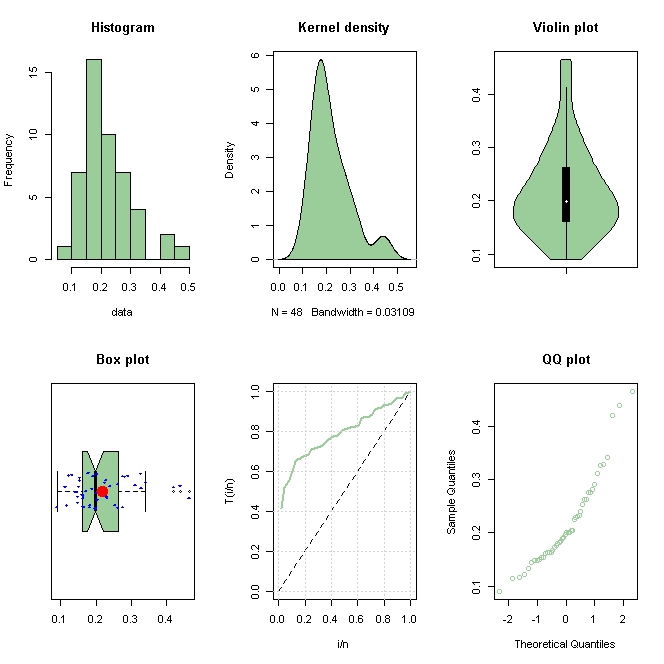}
    \caption{Non-parametric plots for petroleum rock data}
    \label{Fig: nonplot2}
\end{figure}

\begin{table}[H]
\centering
\caption{The comparison statistics, MLEs and SEs of the parameters for the petroleum rock data}%
\label{tab:real2}
\scalebox{0.6} {\ 
\begin{tabular}{cccccccccccccccc}
\hline
Model & $-2\log \ell$ & KS     & AD     & CvM    & p-value(KS) & p-value(AD) & p-value(CvM) & $\hat{\gamma}$    & $\hat{\upsilon}$  & $\hat{p}$       & $\hat{\lambda}$       & SE($\hat{\gamma}$) & SE($\hat{\upsilon}$) & SE($\hat{p}$) & SE($\hat{\lambda}$)  \\\hline
RBTLL & -116.1956 & 0.0672 & 0.1573 & 0.0229 & 0.9818      & 0.9980      & 0.9939       & 9.2012   & 4.7990  & 0.9681  & -                          & 1.2830    & 0.9381      & 0.2899   & -          \\
LL    & -114.9396 & 0.0877 & 0.2874 & 0.0437 & 0.8543      & 0.9470      & 0.9153       & 0.2016   & 4.9490  & -       & -                          & 0.0103    & 0.5932      & -        & -          \\
ELL   & -116.6510 & 0.0834 & 0.1872 & 0.0287 & 0.8919      & 0.9936      & 0.9811       & 115.9440 & 64.9723 & 15.9957 & 0.7352                     & 1891.7200 & 217.0566    & 407.1841 & 0.7056     \\
TW    & -107.8919 & 0.1407 & 0.9700 & 0.1480 & 0.2978      & 0.3729      & 0.3969       & 3.0076   & 0.2797  & 0.6464  & -                          & 0.3112    & 0.0214      & 0.2712   & -          \\
W     & -105.4834 & 0.1499 & 1.2293 & 0.1945 & 0.2310      & 0.2565      & 0.2789       & 2.7475   & 0.2452  & -       & -                          & 0.2844    & 0.0137      & -        & -       \\\hline 
\end{tabular}
}
\end{table}

\begin{figure}[H]
\centerline{\includegraphics[width=4.96in,height=2.82in]{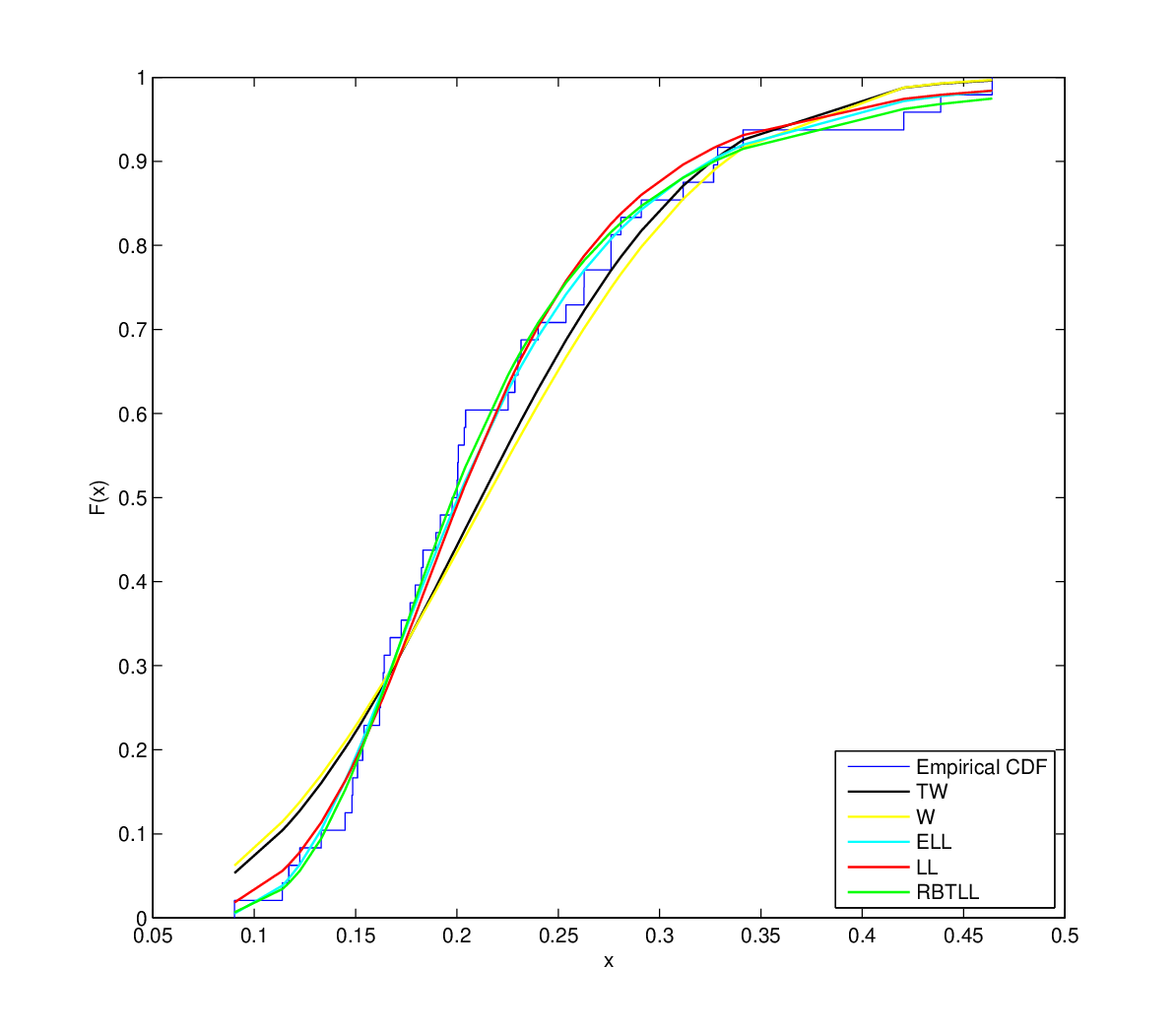}}
\caption{The fitted CDFs for the petroleum rock data set}
\label{fig:real2cdf}
\end{figure}

\begin{figure}[H]
\centerline{\includegraphics[width=4.89in,height=2.78in]{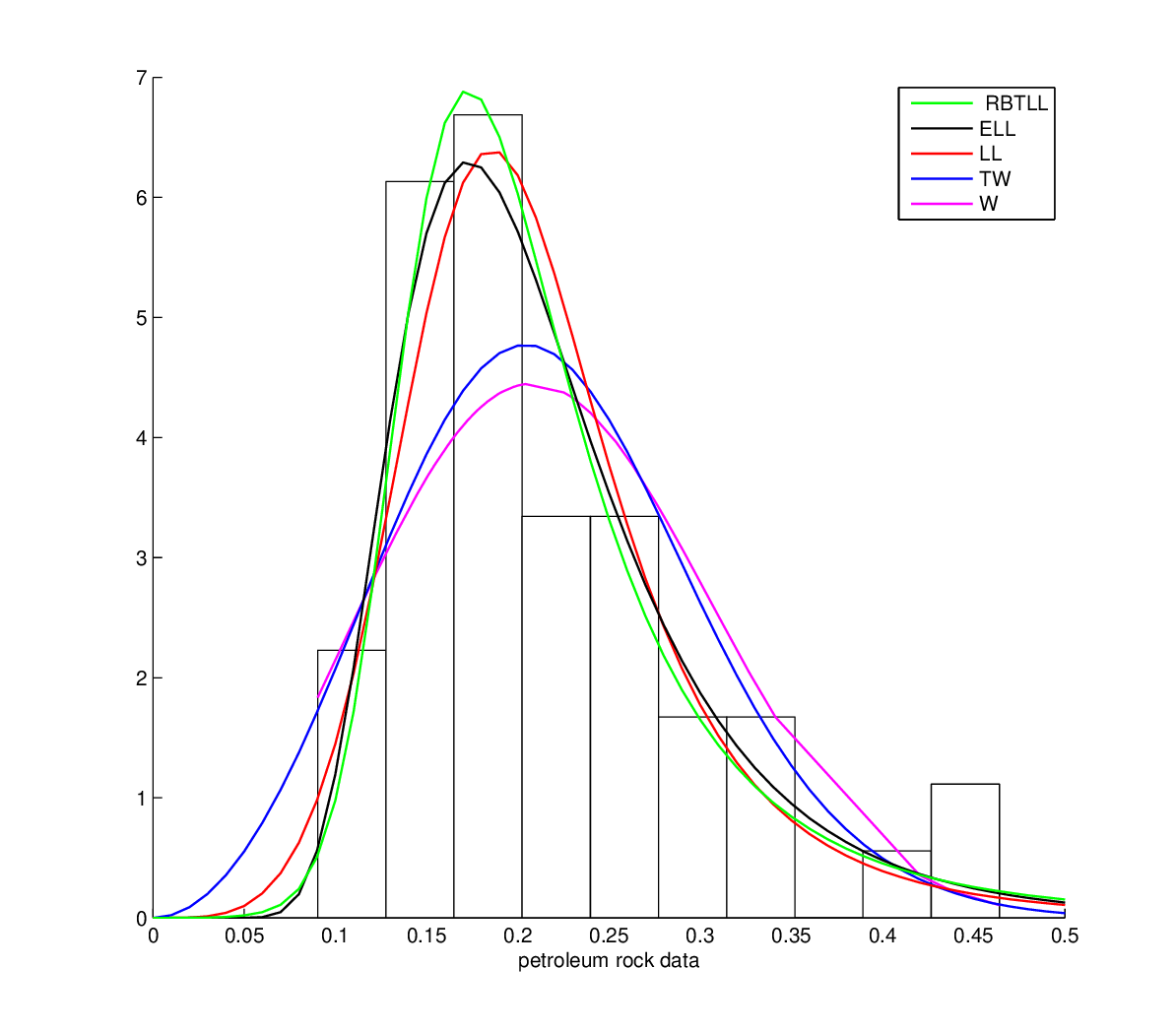}}
\caption{The fitted PDFs for the petroleum rock data set}
\label{fig:real2pdf}
\end{figure}

According to Tables \ref{tab:real1} and \ref{tab:real2}, we conclude that the RBTLL distribution fits the reactor pump failure and petroleum rock data better than competing distributions such as LL, ELL, W and TW according to all the comparison criteria examined.

\section{Conclusion Remarks}
This paper develops a new sub-model of record-based transmuted family of distributions. Some distributional properties of this submodel are examined, the parameter estimation problem is addressed with seven different methods, and two real data studies are presented to compare its data modeling capability with some competing distributions and to prove its usefulness in the analysis of real life data. The results of this study show that the proposed RBTLL distribution fits the two sample data better than the base distribution, the LL distribution, according to all comparison criteria. It should also be noted that the RBTLL distribution is not only capable of modeling the two real datasets given in this study, but these datasets are only presented as examples. In the future, the studies can be planned on modeling the RBTLL distribution for data sets in many fields such as physics, chemistry, biology, health sciences and social sciences and estimating its parameters with different estimation methods or extending this family of distributions by introducing new submodels to the record-based transmuted family of distributions.
\label{Sec: Conc}

\bigskip
\noindent\textbf{Funding} We have not any financial support in this paper.
\newline
\bigskip
\noindent\textbf{Data Availability} The datasets are open access are given in reference list.\\
\noindent\textbf{Conflict of interest} Not Available


\bibliographystyle{apalike} 
\bibliography{ref}
\end{document}